\documentclass[11pt]{article}
\usepackage{xinyu}
\newcommand{\Dodd}[1]{\bD_{\textnormal{odd},\,#1}}
\newcommand{\Deven}[1]{\bD_{\textnormal{even},\,#1}}
\newcommand{\tDodd}[1]{\widetilde{\bD}_{\textnormal{odd},\,#1}}
\newcommand{\tDeven}[1]{\widetilde{\bD}_{\textnormal{even},\,#1}}

\renewcommand{\H}{\operatorname{H}\!}

\title{Analyzing XOR-Forrelation through stochastic calculus}
\author{Xinyu Wu\thanks{Computer Science Department, Carnegie Mellon University. \texttt{xinyuwu@cmu.edu}. Supported by NSF grant CCF-1717606. This material is based upon work supported by the National Science Foundation under grant numbers listed above. Any opinions, findings and conclusions or recommendations expressed in this material are those of the author and do not necessarily reflect the views of the National Science Foundation (NSF).}
}
\date{September 6, 2021}

\begin{document}
\maketitle
\begin{abstract}
    In this note we present a simplified analysis of the quantum and classical complexity of the $k$-XOR Forrelation problem (introduced in the paper of Girish, Raz and Zhan~\cite{GRZ20}) by a stochastic interpretation of the Forrelation distribution.
\end{abstract}

\section{Introduction}
The Forrelation problem~\cite{AA15} and variants of it have been useful in
producing problems that are efficiently solvable by quantum protocols but are hard for classical protocols, in various different models.
A recent line of work analyzing the Forrelation distribution builds on the
polarizing random walk framework introduced by Chattopadhyay, Hatami, Hosseini and Lovett~\cite{CHHL19}.
This framework views the Forrelation distribution as being generated by a random walk in $\R^N$, producing a particular Gaussian distribution, and then rounded to the Boolean cube $\{-1,1\}^N$.
This approach lead to breakthroughs as Raz and Tal's result on the oracle separation of BQP and PH~\cite{RT19} and Bansal and Sinha's proof that $k$-Forrelation exhibits an optimal separation between quantum and classical query complexity~\cite{BS20}\footnote{The proof is phrased in terms of Gaussian interpolation, which is a different viewpoint on the stochastic approach.}.

The recent work of Girish, Raz and Zhan~\cite{GRZ20} analyzes the XOR of $k$ copies of the Forrelation function, and shows that the resulting problem is such that
classical protocols of quasipolynomial size can only achieve quasipolynomially small advantage over random guessing, while there exist quantum protocols with complexity $\polylog(N)$.
They show this for quantum simultaneous-message communication protocols vs.\ classical randomized communication protocols, as well as for quantum query complexity vs.\ classical query complexity.

\vspace*{-1ex}
\paragraph*{Stochastic calculus viewpoint.}
The approach here generalizes~\cite{Wu20} (indeed, the $k=1$ case is identical).
There are two main points where the stochastic approach simplifies the argument in~\cite{GRZ20}.

First,
the Forrelation distribution, prior to rounding, is a truncated multivariate Gaussian.
A multivariate $N$-dimensional Gaussian can also be realized as an $N$-dimensional Brownian motion, stopped at some constant time.
Using a continuous-time random walk allows us to apply stochastic calculus techniques to bound how well $f$ distinguishes the two distributions directly using the $2k\th$ order derivatives of $f$, without the need for additional intermediate bounds.
Furthermore, the Brownian motion approach allows for an induction on $k$, eliminating the need for complex dimension-dependent bounds.

Second, viewing the Gaussian as a Brownian motion also allows us to use a stopping time to encode the truncation. This allows us to directly encode the boundedness of the distribution in the random variable.
This eliminates the
extra step to truncate the Gaussian and bound the closeness in expectation between the truncated and non-truncated Gaussians.

\vspace*{-1ex}
\paragraph*{Connections and future work.}
Conceptually, viewing a truncated Gaussian as a stopped Brownian motion enforces a pathwise view of the random variable, i.e.\ sampling from the distribution means sampling a path of a random walk.
This makes calculations on the distributions easier, for instance because the paths naturally split into ``paths which always remain within the region'' and ``paths which end by hitting the boundary''.
This technique may also be interesting for other applications using truncated Gaussians (or analogously, replacing a truncated exponential distribution by a stopped geometric Brownian motion.)
The stochastic calculus view of Gaussians has also been useful for other Boolean analysis results, for instance in the proof of Bobkov's Two Point Inequality by Barthe and Maurey~\cite{BM00}. Ideas related to the pathwise view of random variables also appear in the recent paper of Eldan and Gross~\cite{EG20}, which expresses the variance and influence of a Boolean function in terms of its action on a certain Brownian motion.

\section{Preliminaries}
We state the main stochastic calculus result we will need in the proof.
This is Dynkin's formula~\cite[Theorem 7.4.1]{Oks03} specialized to our scenario of the Brownian motion having mean 0 and constant covariance.
\begin{theorem}\label{thm:dynkin}
    Let $\bX$ be an $n$-dimensional Brownian motion with mean $0$ and covariance $\Sigma$, let $\btau$ be a bounded stopping time, and let $f:\R^N \to \R$ be a twice continuously differentiable function.
    We use $\H f$ to denote the Hessian of $f$, the $N \times N$ matrix of second order partial derivatives.
    The following holds:
    \[
        \E[f(\bX_\btau)] = f(0) + \E\bracks*{\int_0^\btau \frac12\angle{\Sigma,\, \H f(\bX_s)}\, ds}.
    \]
\end{theorem}

We also need the following formula regarding random restrictions, which is essentially Lemma~1 of~\cite{Wu20}. A similar idea appears in the proof of Lemma 5.1 of~\cite{GRZ20}, and previously in~\cite[Claim A.5]{CHLT18}.
\begin{lemma}\label{lem:derivs}
Let $f:\R^N \to \R$ be a multilinear polynomial. For any $x \in [-1/2,1/2]^N$, there exists a distribution $\calR_x$ over restrictions $\brho \in \{-1,1,*\}^N$, such that for any $S \subseteq [N]$,
\[
     \pt_{S}f(x) = 2^{|S|}\E_{\brho \sim \calR_x}\bracks*{\pt_{S}f_{\brho}(0)}.
\]
Here we write $\pt_{S} = \prod_{i\in S}\frac{\pt}{\pt_i}$ for the partial derivatives over the coordinates in $S$. We further define the \emph{Fourier coefficient} $\wh f(S) \coloneqq \pt_S f(0)$. Note that this coincides with the usual decomposition $f(x) = \sum_{S \subseteq [n]} \wh{f}(S) \prod_{i \in S} x_i$.
\end{lemma}

\section{A bound on a product of Brownian motions}

\begin{definition}\label{def:setup}
Let $k \in \N_+$, and let $\bX^{(1)},\dots, \bX^{(k)}$ be identical independent $N$-dimensional Brownian motions with mean 0 and covariance matrix $\Sigma$, and let $\btau_1, \dots, \btau_k$ be stopping times.

We consider distributions on $\R^{kN} \cong (\R^N)^k$, which we take to be $k$ copies of $\R^N$, indexed by coordinates $1,\dots,k$.
Let $S \subseteq [k]$.
Define the random variable $\bX^S_\btau$ to be
$\bX^{(i)}_{\btau_i}$ in the $i\th$ coordinate if $i \in S$, and $0$ in the $i\th$ coordinate if $i \notin S$.
We set $\bD_S$ to be the distribution of $\bX^S_\btau$.

We write $\bS \sim [k]$ to denote drawing $\bS \subseteq [k]$ uniformly.
We now define the distribution $\Dodd{k}$ to be the distribution of $\bD_\bS$ conditioned on $|\bS|$ being odd.
Similarly, we define $\Deven{k}$ to be $\bD_\bS$ conditioned on $|\bS|$ being even.
When $k=1$, we define $\bD_1 = \bX^{(1)}_{\btau_1} = \Dodd{1}$.

For a multilinear function $f:\R^{kN} \to \R$, we note the identity
\begin{equation}\label{eqn:difference}
    \E[f(\Deven{k})] - \E[f(\Dodd{k})] = 2 \E_{\bS \sim [k]}\bracks*{(-1)^{|\bS|} f( \bD_\bS)}.
\end{equation}
\end{definition}

The following bounds how well a Boolean function with bounded level-$2k$ Fourier weight can distinguish $\Deven{k}$ and $\Dodd{k}$.
This is essentially Theorem 3.1 of~\cite{GRZ20}.
\begin{theorem}\label{thm:main}
Let $k \in \N_+$, let $f:\{-1,1\}^{kN} \to \{-1,1\}$ be a Boolean function, and let $L > 0$ be such that for any restriction~$\rho$,
\[
    \sum_{\substack{S \se [kN] \\ |S| = 2k}} |\wh{f_\rho}(S)| \leq L.
\]
Let $\gamma > 0$ and let $\bX^{(1)},\dots, \bX^{(k)}$ be identical independent $N$-dimensional Brownian motions with mean 0 and covariance matrix $\Sigma$.
Further assume that $|\Sigma_{ij}| \leq \gamma$ for $i \ne j$.

Let $\ep > 0$ and define the (bounded) stopping times for each $i \in [k]$,
\[
    \btau_i \coloneqq \min\,\{\ep, \text{ first time that $\bX^{(i)}$ exits } [-1/2,1/2]^N\}.
\]

Then, identifying $f$ with its multilinear expansion, we have
\begin{equation*}
    \abs{\E_{\bS \sim [k]}\bracks*{(-1)^{|\bS|} f( \bD_\bS)}}\leq (\ep \gamma)^k L.
\end{equation*}
\end{theorem}
\begin{proof}
We first prove by induction on $k$ that for any multilinear function $f$,
\begin{equation} \label{eqn:induction}
    \E_{\bS \sim [k]}\bracks*{(-1)^{|\bS|}f(\bD_\bS)} = \E \bracks*{\int_0^{\btau_1}\cdots\int_0^{\btau_k} \frac{(-1)^{k-1}}{2^{2k-1}} \angle*{(I_k\otimes \Sigma)^{\otimes k},\; \H^{\otimes k} f(\bX_{t_1}^{(1)}, \dots, \bX_{t_k}^{(k)})}\, dt_1\dots d t_k},
\end{equation}
where $\H^{\otimes k} f$ denotes the $(kN)^k$-dimensional matrix of all the $2k\th$ order derivatives of $f$, $I_k$ is the $k$-dimensional identity matrix, and $\otimes$ denotes the Kronecker product.

The base case $k = 1$ is simply a direct application of Dynkin's formula (\Cref{thm:dynkin}):
\[
    \E[f(\bX_{\btau_1}^{(1)})] - f(0) = \E\bracks*{\int_0^{\btau_1} \frac12 \angle*{\Sigma,\,\H f(\bX_{t_1}^{(1)})}\, dt_1}.
\]

For the induction step, we condition on the last coordinate to observe that
\begin{equation}\label{eqn:coord-k}
    \E_{\bS \sim [k]}\bracks*{(-1)^{|\bS|}f(\bD_\bS)} = \frac12 \E_{\bS \sim [k-1]} \bracks*{(-1)^{|\bS|}f(\bD_\bS, 0)} - \frac12\E_{\bS \sim [k-1]} \bracks*{(-1)^{|\bS|}f(\bD_\bS, \bX_{\btau_k}^{(k)})}
\end{equation}
Now let $\bS \sim [k-1]$.
We will proceed by applying Dynkin's formula to $g(x) = \E_{\bS,\,\bD_\bS}[(-1)^{|\bS|}f(\bD_\bS, x)]$.
Since $f$ is a multilinear function, partial derivatives of $g$ commute with the expectation; in particular $\partial_i g(x) = \E[(-1)^{|\bS|}\partial_{i+(k-1)N} f(\bD_\bS,x)]$, and so, with $e_k \in \R^k$ denoting the indicator of the $k\th$ coordinate,
\begin{align}
    \E_{\bS,\,\bD_\bS,\,\bX^{(k)},\,\btau_k}[& (-1)^{|\bS|}f(\bD_\bS,\bX_{\btau_k}^{(k)})] - \E_{\bS,\,\bD_\bS}[(-1)^{|\bS|} f(\bD_\bS,0)] \notag\\
    &= \E_{\bX^{(k)},\,\btau_k}\bracks*{\frac12\int_0^{\btau_k} \angle*{e_k e_k^T \otimes \Sigma, \E_{\bS,\,\bD_\bS}\bracks*{(-1)^{|\bS|}\H f(\bD_\bS, \bX_{t_k}^{(k)})}}\, dt_k} \label{eqn:induct}
\end{align}
Finally, we apply the induction hypothesis to find, for $(k-1)N < i, j \leq kN$,
\begin{align*}
\hspace*{-1em}
    \E_{\bS \sim [k-1]}&\bracks*{(-1)^{|\bS|}\partial_{i,j} f(\bD_\bS, \bX_{t_k}^{(k)})} \\
    &= \E \bracks*{\int_0^{\btau_1}\dots\int_0^{\btau_{k-1}} \frac{(-1)^{k-2}}{2^{2k-3}} \angle*{(I_{k-1} \otimes \Sigma)^{\otimes (k-1)},\; \H^{\otimes (k-1)}_{\bX^{(1)},\dots,\bX^{(k-1)}} \partial_{i,j}f(\bX_{t_1}^{(1)}, \dots, \bX_{t_{k}}^{(k)})}\, dt_1\dots d t_{k-1}}.
\end{align*}
Combining with \Cref{eqn:coord-k,eqn:induct} and using bilinearity of the inner product, we conclude
\begin{align*}
\hspace*{-3em}
    \E_{\bS \sim [k]}&\bracks*{(-1)^{|\bS|}f(\bD_\bS)} \\
    &= -\frac12\E_{\bX^{(k)},\,\btau_k}\bracks*{\frac12\int_0^{\btau_k} \angle*{e_k e_k^T \otimes \Sigma, \E_{\bS,\,\bD_\bS}\bracks*{(-1)^{|\bS|}\H f(\bD_\bS, \bX_{t_k}^{(k)})}}\, dt_k}\\
    &= \E\bracks*{\int_0^{\btau_1}\dots\int_0^{\btau_{k}} \frac{(-1)^{k-1}}{2^{2k-1}} \angle*{e_ke_k^T \otimes \Sigma, \H_{\bX^{(k)}} \angle*{(I_{k-1} \otimes \Sigma)^{\otimes (k-1)},\,\H^{\otimes (k-1)}_{\bX^{(1)},\dots,\bX^{(k-1)}} f(\bX_{t_1}^{(1)}, \dots, \bX_{t_{k}}^{(k)})}}\, dt_1\dots d t_{k}}\\
    &= \E \bracks*{\int_0^{\btau_1}\dots\int_0^{\btau_k} \frac{(-1)^{k-1}}{2^{2k-1}} \angle*{(I_k\otimes \Sigma)^{\otimes k},\; \H^{\otimes k} f(\bX_{t_1}^{(1)}, \dots, \bX_{t_k}^{(k)})}\, dt_1\dots d t_k}.
\end{align*}

Having completed the proof of~\Cref{eqn:induction}, we now use it to prove the theorem
\begin{align*}
    \E_{\bS \sim [k]}&\bracks*{(-1)^{|\bS|}f(\bD_\bS)} \\
    &\leq \eps^k \E\bracks*{\sup_{\substack{t_1 \in [0, \btau_1] \\ \dots \\ t_k \in [0,\btau_k]}} \abs{\frac{1}{2^{2k}} \angle*{(I_k\otimes \Sigma)^{\otimes k},\; \H^{\otimes k} f(\bX_{t_1}^{(1)}, \dots, \bX_{t_k}^{(k)})} }}
    &&(\btau_1,\dots,\btau_k \leq \eps)\\
    &\leq \frac{(\eps\gamma)^k}{2^{2k}} \sup_{(x_1,\dots,x_k) \in [-1/2, 1/2]^{kN} } \sum_{\substack{S \subseteq [kN] \\ |S| = 2k}} \abs{\pt_S f(x_1,\dots,x_k)}
    &&(|\Sigma_{ij}| \leq \gamma \text{ for } i \neq j,\; \pt_{ii}f = 0)\\
    &\leq (\eps\gamma)^k \sup_{(x_1,\dots,x_k) \in [-1/2, 1/2]^{kN} } \sum_{\substack{S \subseteq [kN] \\ |S| = 2k}} \abs{\E_{\brho \sim \calR_{x_1,\dots,x_k}}[\pt_S f_\rho(0,\dots,0)]}
    &&(\text{\Cref{lem:derivs}})\\
    &\leq (\eps\gamma)^k \sup_{(x_1,\dots,x_k) \in [-1/2, 1/2]^{kN} } \E_{\brho \sim \calR_{x_1,\dots,x_k}}\bracks*{\sum_{\substack{S \subseteq [kN] \\ |S| = 2k}} \abs{ \wh f_\rho(S)}}&&\\
    &\leq (\eps\gamma)^kL.&&\qedhere
\end{align*}
\end{proof}

\section{Application to complexity of \texorpdfstring{$k$}{k}-XOR Forrelation}
We now apply the bound from the previous section to prove the main theorem from~\cite[Theorem 3.1]{GRZ20}, from which they derive separations in quantum versus classical query complexity, communication complexity and circuit complexity (we refer to~\cite{GRZ20} for the exact details about the definitions of the complexity classes and the full proof).

We briefly sketch how the proof in~\cite{GRZ20} proceeds.
For the lower bounds on the classical complexity classes, it suffices to exhibit two distributions that are hard for functions in the complexity class to distinguish. These will be derived from $\Dodd{k}$ and $\Deven{k}$, with $\Sigma$ and $\eps$ chosen appropriately (\cite{GRZ20} uses the truncated Gaussian instead of the stopping time here).
\cite{GRZ20} observes that these classical complexity classes are closed under restrictions.
Then,~\Cref{thm:main,eqn:difference} combined with bounds on the level-$k$ Fourier weights proven in~\cite{Tal19} and~\cite{GRZ20} shows the classical lower bounds.

For the quantum upper bound, we need to show that a quantum query algorithm (or communication protocol respectively) can distinguish $\Dodd{k}$ and $\Deven{k}$ with high probability.
We will show that the concentration results proven in \cite{GRZ20} hold in our context as well.

We now set values for $\eps$ and $k$. We take $\eps = 1/(28 k^2 \ln N)$, and $k$ small enough that $\eps^2N \leq \poly(N)$ (e.g. $k \leq O(N^{1/5})$ suffices), and set
\[
    \Sigma \coloneqq \begin{pmatrix}
    I_n & H_n\\ H_n & I_n
    \end{pmatrix},
\]
where $N = 2n$, $n$ is a power of $2$ and $H_n$ is the normalized Hadamard matrix, so $\gamma = \frac{1}{\sqrt{n}}$.
Applying~\Cref{thm:main}, the overall upper bound is $L_{2k}\cdot\text{polylog}(N)/N^{k}$, where $L_{2k}$ is the bound on the Fourier weight at level $2k$ for the family of functions in the complexity class.

The quantum algorithm/communication protocol is based on the $k$-XOR Forrelation problem, which we define here:
Let $\phi:\R^n \times \R^n \to \R$ as $\phi(x,y) \coloneqq \frac 1n \angle{x,H_ny}$. The Forrelation decision problem is a partial function defined by
\begin{align*}
    F(x,y) = \begin{cases}
    -1 &\text{ if } \phi(x,y) \geq \eps/2,\\
     1 &\text{ if } \phi(x,y) \leq \eps/4.
    \end{cases}
\end{align*}
The $k$-XOR Forrelation $F^{(k)}:\{-1,1\}^{kN} \to \{-1,1\}$ is defined by $F^{(k)}(z_1,\dots,z_k) = \prod_{i=1}^k F(z_i)$.

Since $\Dodd{k}$ and $\Deven{k}$ take values in $[-1/2,1/2]^{kN}$ but $F^{(k)}$ is defined on $\{-1, 1\}^{kN}$, we round them to distributions $\tDodd{k}$ and $\tDeven{k}$ on $\{-1,1\}^{kN}$.
A draw of $\td\bz \sim \tDodd{k}$ is defined as follows:
\begin{enumerate}[nosep]
    \item Sample $\bz \sim \Dodd{k}$.
    \item For each coordinate $i \in [N]$, independently set $\widetilde{\bz}_i = 1$ with probability $\frac{1 + \bz_i}{2}$ and $-1$ with probability $\frac{1-\bz_i}{2}$. We denote $\td\bz \sim \bz$ for this step. Now $\td\bz$ is sampled from $\tDodd{k}$.
\end{enumerate}
$\tDeven{k}$ is defined analogously.
Note that for a multilinear polynomial $f:\R^N \to \R$, $\E[f(\Dodd{k})] = \E[f(\tDodd{k})]$ (analogously for $\tDeven{k}$).

Girish, Raz and Zhan showed the following about the rounding process:
\begin{proposition}[Claim~A.2~\cite{GRZ20}]\label{prop:rounding}
    Let $z \in [-1/2,1/2]$, and let $\td\bz \sim z$ as in step 2 above. Then,
    \[
        \P[|\phi(\td\bz) - \phi(z)| > \eps/4] \leq \exp(-\Omega (N^{1/4})).
    \]
\end{proposition}

Finally, we show a concentration result analogous to Lemma~2.11 of~\cite{GRZ20} which shows that $F^{(k)}$ decides correctly on $\tDodd{k}$ and $\tDeven{k}$ with high probability.
This is then sufficient to deduce the applications described in~\cite{GRZ20}.

First, we prove a concentration bound for $\phi(\bD_1)$, i.e.\ $(\bx, \by)$ are generated by a single $N$-dimensional stopped Brownian motion with covariance $\Sigma$.
\begin{lemma}\label{lem:X1}
In the above context, the following holds:
\begin{equation}\label{eqn:forr-bound-walk}
    \P_{(\bx,\by)\sim \bD_1}[\phi(\bx,\by) \geq 3\eps/4] \geq 1-O(1/N^{6k^2}).
\end{equation}
\end{lemma}
\begin{proof}
Notice that an alternate way to sample $(\bx, \by) \sim \bD_1$ is to let $\bX_t$ be a $n$-dimensional Brownian motion with covariance $I_n$ stopped at the stopping time
\[
    \btau \coloneqq \min\{\eps,\, \text{first time that $\bX_t$ or $H_n \bX_t$ exits $[-1/2,1/2]^n$}\},
\]
and let $(\bx,\by) = (\bX_\btau, H_n \bX_\btau)$. Then, $\phi(\bx,\by) = \frac1n\norm{\bX_\btau}_2^2$.
In order to prove the desired bound, we first prove that with high probability $\btau = \eps$, i.e.\ the path of the Brownian motion did not exit $[-1/2, 1/2]^N$ before time $\eps$.
We then show that $\frac1n \norm{\bX_\eps}_2^2 \geq 3\eps/4$ with high probability, and conclude using a union bound.
We can union bound over the $N$ coordinates,
\[
    \Pr[\btau < \eps]
    \leq N \cdot \P[\text{1st coordinate of $X_t$ exits $[-1/2,1/2]$ earlier than $\eps/2$}].
\]
Since each coordinate of $\bX_t$ is a standard 1D Brownian motion $\bB_t$,
we can apply Doob's submartingale inequality (e.g.~\cite[Proposition II.1.8]{RY99}) to obtain
\[
    \Pr\bracks*{\sup_{0 \leq t \leq \ep/2} |\bB_t| \geq \frac12} \leq 2e^{-1/4\ep} = 2e^{-7k^2\ln N} = \frac{2}{N^{7k^2}}.
\]
Therefore,
\begin{equation}\label{eqn:bm-stops-early}
\Pr[\btau < \eps] \leq 2/N^{7k^2-1}.
\end{equation}

Next, we consider $\frac1n \norm{\bX_\eps}_2^2$. Note that this is simply the average of the squares of $n$ iid Gaussians $\bx_1, \dots, \bx_n$ with mean $0$ and variance $\eps$. Using \cite[Example 2.11]{Wai19}, we have the tail bound
\begin{equation}\label{eqn:chi-squared}
    \Pr\bracks*{\abs{\frac{1}{n} \sum_{i=1}^n \bx_i^2 - \eps} \geq \frac\eps4} \leq \exp(-\Omega(N)).
\end{equation}

Taking a union bound over \Cref{eqn:bm-stops-early,eqn:chi-squared}, we have $\P_{(\bx,\by)\sim\bD_1}[\phi(\bx,\by) \leq 3\eps/4] \leq O(1/N^{6k^2})$.
\end{proof}

\begin{proposition}\label{prop:concentration}
    The following hold:
    \[
        \P_{\td\bz \sim \tDeven{k}}[F^{(k)}(\td\bz) = 1] \geq 1-O\parens*{\frac{k}{N^{6k^2}}} \qquad\text{and}\qquad
        \P_{\td\bz \sim \tDodd{k}}[F^{(k)}(\td\bz) = -1] \geq 1-O\parens*{\frac{k}{N^{6k^2}}}.
    \]
\end{proposition}
\begin{proof}
We first show $F$ decides correctly on the coordinates with $\bU_N$ with high probability:
\begin{equation}\label{eqn:forr-bound-unif}
    \P_{(\bx,\by)\sim \bU_N}[\phi(\bx,\by) \leq \eps/4] \geq 1-\exp(-\Omega(N\eps^2)).
\end{equation}
To see this, note that $\bx$ and $\by$ are independent, so $\bx$ and $H_n\by$ are independent. Hence $\phi(\bx,\by)$ is simply the average of random signs, and the bound holds by Hoeffding's inequality.

Finally, to prove the proposition it suffices to prove that for any fixed $S \subseteq [k]$, $\P_{\bz \sim \bD_S,\, \td\bz \sim \bz}[F^{(k)}(\td\bz) \neq (-1)^{|S|}] \leq O(k/N^{6k^2})$. If $i \in S$, then $\td\bz_i$ is distributed as $\td\bD_1$, so \Cref{lem:X1} combined with \Cref{prop:rounding} implies $\P[F(\td\bz_i) = 1] \leq O(1/N^{6k^2})$. Meanwhile if $i \notin S$, then $\td\bz_i$ is distributed as $\bU_N$, so \Cref{eqn:forr-bound-unif} implies $\P[F(\td\bz_i) = -1] \leq \exp(-\Omega(N\eps^2))$. With $k$ and therefore $\eps$ taken sufficiently small, a union bound over the $k$ coordinates completes the proof.
\end{proof}

\section*{Acknowledgements}
I would like to thank Ryan O'Donnell for many helpful comments on this paper.

\end{document}